\documentclass[conference]{IEEEtran}
\IEEEoverridecommandlockouts
\usepackage{makecell}
\usepackage{cite}
\usepackage{amsmath,amssymb,amsfonts}
\DeclareMathOperator*{\argmax}{arg\,max}
\usepackage{algorithmic}
\usepackage{graphicx}
\usepackage{textcomp}
\usepackage{xcolor}

\newtheorem{lemma}{Lemma}
\newtheorem{remark}{Remark}
\usepackage{bbm}

\usepackage{bm}
\newtheorem{proof}{Proof}
\def\BibTeX{{\rm B\kern-.05em{\sc i\kern-.025em b}\kern-.08em
    T\kern-.1667em\lower.7ex\hbox{E}\kern-.125emX}}
\begin{document}

\setlength{\textfloatsep}{0.11cm}
\setlength{\abovedisplayskip}{0.10cm}
 \setlength{\belowdisplayskip}{0.10cm}
 \setlength{\abovecaptionskip}{-1mm}

\title{MIMO Asynchronous MAC with Faster-than-Nyquist (FTN)  Signaling\\
\thanks{This work was funded in part by a
Discovery Grant awarded by the Natural Sciences and Engineering
Research Council of Canada (NSERC), 
 in part by the Scientific and Technological Research Council of Turkey, TUBITAK, under grant
122E248 and in part by the Research Fund of the Middle
East Technical University, under project 11104.}
}


\author{Zichao~Zhang\textsuperscript{\P} ,
		~Melda~Yuksel\textsuperscript{\S},~Halim~Yanikomeroglu\textsuperscript{\P},~Benjamin~K.~Ng\textsuperscript{\dag},~Chan-Tong~Lam\textsuperscript{\dag}  \\
  \textsuperscript{\P}Department of  Systems and Computer Engineering, Carleton University, Ottawa, ON, Canada  \\
  \textsuperscript{\S}Department of Electrical and Electronics Engineering, Middle East Technical University, Ankara, Turkey  \\
  \textsuperscript{\dag}Faculty of Applied Sciences, Macao Polytechnic
 University, Macao SAR, China \\ \\
 Emails: zichaozhang@cmail.carleton.ca, ymelda@metu.edu.tr, halim@sce.carleton.ca, \{bng, ctlam\}@mpu.edu.mo 
}

\maketitle

\begin{abstract}

Faster-than-Nyquist (FTN) signaling is a non-orthogonal transmission technique, which brings in intentional inter-symbol interference. This way it can significantly enhance spectral efficiency for practical pulse shapes such as the root raised cosine pulses. This paper proposes an achievable rate region for the multiple antenna (MIMO) asynchronous multiple access channel (aMAC) with FTN signaling. The scheme applies waterfilling in the spatial domain and precoding in time. Waterfilling in space provides better power allocation and precoding helps mitigate inter-symbol interference due to asynchronous transmission and FTN. The results show that the gains due to asynchronous transmission and FTN are more emphasized in MIMO aMAC than in single antenna aMAC. Moreover, FTN improves single-user rates, and asynchronous transmission improves the sum-rate, due to better inter-user interference management.   
\end{abstract}


\section{Introduction}
Multiple-input multiple-output (MIMO) transmission schemes have gained an increasing momentum over the last 20 years \cite{telatarmimo,massivemimo,cellfree}. MIMO is vitally important in 5G and beyond since bandwidth resources are limited and the number of devices keep growing \cite{5g}. MIMO transmission offers a solution to this problem as it provides both spatial multiplexing and diversity gains that drastically improve transmission rates without needing more bandwidth \cite{telatarmimo}. 

Traditional multiple access technologies, such as time-division multiple access (TDMA) or frequency-division multiple access (FDMA) are not sufficient to accommodate a large number of devices with limited resources. Therefore, alternative approaches, such as non-orthogonal multiple-access (NOMA) technologies for both uplink and downlink, have been proposed to improve spectral efficiency \cite{noma}.  In the uplink, NOMA is a special case of the multiple access channel (MAC). Unlike traditional methods, multiple access transmission allows multiple users to share the same frequency and time resources \cite{noma, cover}. 

The capacity region for the conventional two-user single-input single-output (SISO) MAC is found in \cite{cover}, while its MIMO counterpart is found in \cite{goldsmith}. 
It is shown in \cite{verdu} that asynchronous MAC (aMAC) expands the synchronous MAC capacity regions in \cite{cover} and \cite{goldsmith} further. Asynchronous trasmission, or an intentional extra time delay between users enables better inter-user interference mitigation and is  beneficial to the system. In \cite{cellfree}, the authors study the downlink performance of cell-free asynchronous massive MIMO and then apply rate splitting to deal with multi-user interference.

Another promising technology for future communication systems is faster-than-Nyquist (FTN) signaling \cite{evolutionftn}. FTN improves spectral efficiency at the expense of computational complexity. Due to the increased transmission rate, the transmitted pulses do not satisfy Nyquist criterion and are not orthogonal to each other. This induces inter-symbol interference (ISI). However, this intentional ISI can be controllable under certain conditions \cite{ourpaper}, we can use precoding to detangle the correlated symbols at the receiver. Therefore, it is reasonable to employ MIMO, asynchronous transmission and FTN together to boost the performance.

In this paper, we study the performance of MIMO aMAC with FTN. We first develop the system model in Section \ref{sec:model}. We then propose a power allocation scheme that induces an achievable rate region in Section~\ref{sec:powallo}. We present the achievable regions and sum rate performance in Section~\ref{sec:num}. Finally, we conclude the paper in Section~\ref{sec:conc}.

\section{System Model}\label{sec:model}

The multiple-access communication system is assumed to have two transmitters and one receiver. The transmitters have $L$ antennas each and the receiver is equipped with $M$ antennas. In order to reap the gains due to asynchronous transmission \cite{verdu}, each user transmits with a fixed time delay $\tau_k$, $k=1,2$. We denote the $s$th symbol, $s=0,\ldots,N-1$ transmitted from the $l$th antenna, $l =1,\ldots,L$, of the $k$th user, $k=1,2$, as $a^l_k[s]$ . Both the transmitters use the same pulse shaping filter $p(t)$ and the receiver employs the matched filter $p^*(-t)$. In our paper, we assume $p(t)$ is the root-raised cosine pulse with roll-off factor $\beta$.  Furthermore, there is quasi-static fading, denoted with $h_k^{ml}$, indicating the channel coefficient between the $l$th transmit antenna of the $k$th user to the $m$th receive antenna, $m=1,\ldots,M$. Both users transmit $N$ symbols with symbol period $\delta T$, where $\delta \in [0,1]$ is the acceleration factor in FTN. With the presence of complex additive white Gaussian noise $z(t)$ at the receiver, the signal at the output of the matched filter of the $m$th receive antenna is
\begin{figure}[t]
    \centering
    \includegraphics[scale=0.06]{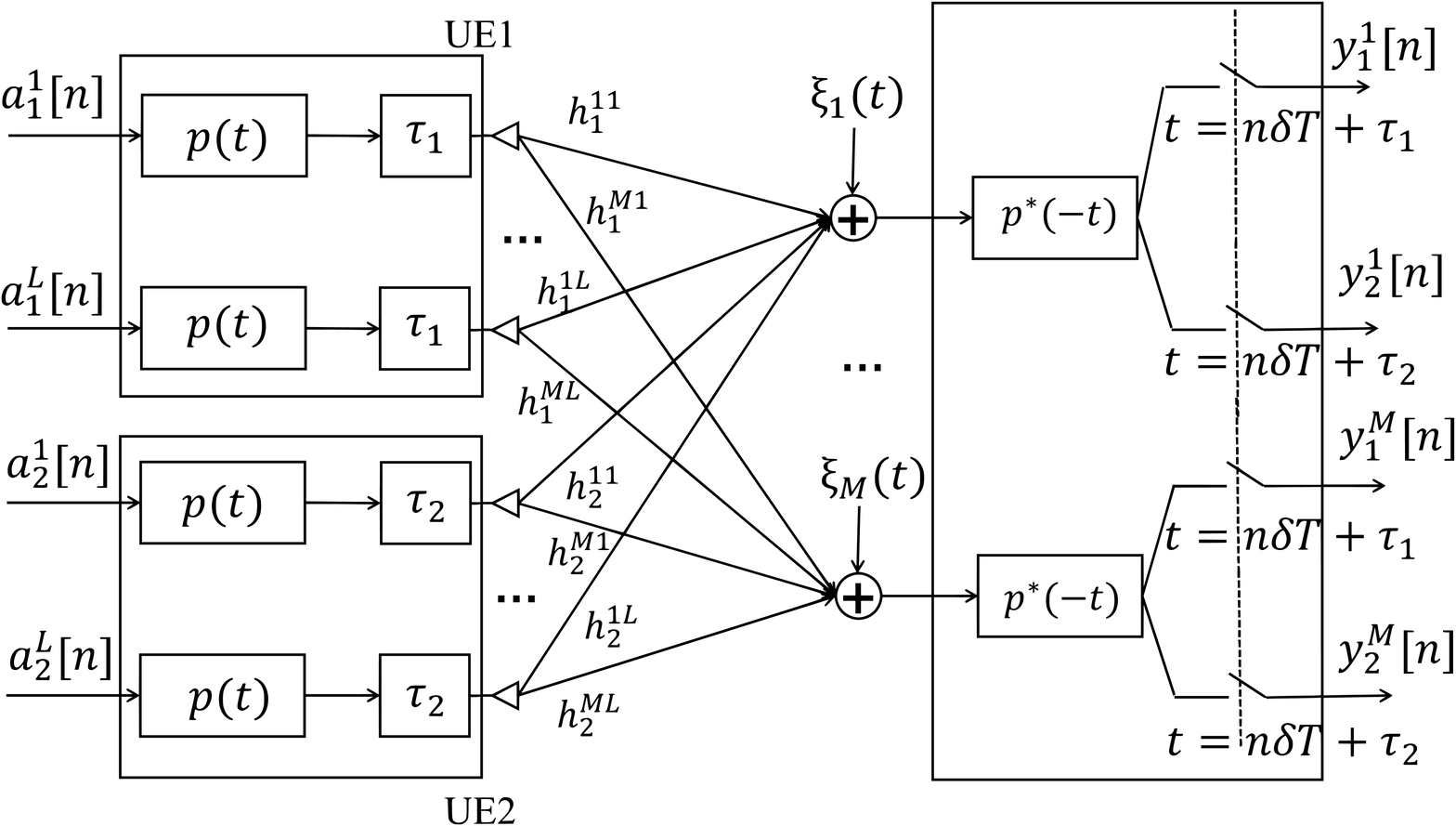}
    \caption{System model for MIMO aMAC FTN.}
    \label{fig:my_label}
\end{figure}
\begin{equation}
y^m(t)=\sum_{k=1}^2\sum_{l=1}^Lh_{k}^{ml}\sum_{s=0}^{N-1} a_k^l[s]g(t-s\delta T-\tau_k) + \eta(t).
\end{equation}
Here $g(t)=p(t)\star {p^*(-t)}$, and $\eta(t)=z(t)\star {p^*(-t)}$, where $\star$ represents the convolution operation. In order to have sufficient statistics \cite{verdu}, we need to sample $y^m(t)$ according to each user's time delay. This requires us to sample at time instants $n\delta T+\tau_{k'}$ to obtain the set of samples $y^m_{k'}[n]$, $k'=1,2$. As a result, there are $2MN$ samples in total, namely $y^m_{k'}[n]=y^m(n\delta T+\tau_{k'})$ where 
\begin{align}
y^m_{k'}[ n]&=\sum_{k=1}^2\sum_{l=1}^Lh_{k}^{ml}\sum_{s=0}^{N-1} a_k^l[s]g((n-s)\delta T+(\tau_{k'}-\tau_k)) \notag\\ &\quad\quad\quad\quad\quad\quad\quad\quad+\eta(n\delta T+\tau_{k'}).
\end{align}
This system model, shown in Fig. \ref{fig:my_label}, can also be written in a matrix multiplication form as
\begin{align}
\left[\begin{matrix}\bm{Y}_1 \\ \bm{Y}_2\end{matrix}\right] 
 &=\left[\begin{matrix}\bm{H}_1\otimes\bm{G}_\delta & \bm{H}_2\otimes\bm{G}_{\delta,12} \\ \bm{H}_1\otimes\bm{G}_{\delta,21} & \bm{H}_2\otimes\bm{G}_\delta\end{matrix}\right]\left[\begin{matrix}\bm{A}_1 \\ \bm{A}_2\end{matrix}\right]  +\left[\begin{matrix}\bm{\Omega}_1 \\ \bm{\Omega}_2\end{matrix}\right],   \\
 \textrm{or}\quad \bm{Y} &=\bm{H}\bm{A}+\bm{\Omega},
\end{align}
where $\bm{Y}_k=[(\bm{y}_k^1)^T,\dots,(\bm{y}_k^M)^T]^T$, $\bm{A}_k=[(\bm{a}_k^1)^T,\dots,(\bm{a}_k^L)^T]^T$ and $\bm{\Omega}_k=[(\bm{\eta}_k^1)^T,\dots,(\bm{\eta}_k^M)^T]^T$, $k=1,2$. We also have $\bm{a}_k^l=[a_k^l[0],a_k^l[1], \dots, a_k^l[N-1]]$,  $\bm{y}_k^l=[y_k^l[0],y_k^l[1], \dots, y_k^l[N-1]]$, $\bm{\eta}_k^l=[\eta_k^l[0],\eta_k^l[1], \dots, \eta_k^l[N-1]]$.  The matrix $\bm{H}_k$ is the $M\times L$ channel coefficient matrix between the $k$th user and the receiver. Its entries are given as $(\bm{H}_k)_{m,l}=h^{ml}_k$. The matrices $\bm{G}_\delta$, $\bm{G}_{\delta,12}$ and $\bm{G}_{\delta,21}$ have their entries equal to $(\bm{G}_{\delta})_{n,s}=g((n-s)\delta T)$, $(\bm{G}_{\delta,12})_{n,s}=g((n-s)\delta T+(\tau_1-\tau_2))$ and similarly $(\bm{G}_{\delta,21})_{n,s}=g((n-s)\delta T+(\tau_2-\tau_1))$. We can see that  $\bm{G}_\delta=\bm{G}_\delta^\dagger$, $\bm{G}_{\delta,12}^\dagger=\bm{G}_{\delta,21}$ and they are all Toeplitz matrices \cite{gray}. Note that an $N\times N$ Toeplitz matrix $\bm{T}_N$ has entries $(\bm{T}_N)_{n,s}=t_{n-s}$, $n,s=0,\dots,N-1$.

\subsection{The Mutual Information Expressions}

The mutual information expression for the sum rate is written as \cite{cover} 
\begin{align}
    I(\bm{Y};\bm{A}_1,\bm{A}_2)&=h(\bm{Y})-h(\bm{Y}|\bm{A}_1,\bm{A}_2) \\
    &=\log\det\left(\mathbb{E}\left[\bm{Y}\bm{Y}^\dagger\right]\right)-\log\det\left(\mathbb{E}\left[\bm{\Omega}\bm{\Omega}^\dagger\right]\right), \label{eqn:summuldef}
\end{align}
where $h(\cdot)$ is the differential entropy. The covariance matrix for the colored Gaussian noise samples $\bm{\Omega}$ is 
\begin{align}
\mathbb{E}\left[\bm{\Omega}\bm{\Omega}^\dagger\right]&\triangleq\bm{\Sigma_\Omega}
=\left[\begin{matrix}\bm{I}_M\otimes\bm{G}_\delta &\bm{I}_M\otimes\bm{G}_{\delta,12} \\ \bm{I}_M\otimes\bm{G}_{\delta,21}  & \bm{I}_M\otimes\bm{G}_\delta \end{matrix}\right],\label{eqn:noisecov}
\end{align} where $\bm{I}_M$ is the $M\times M$ identity matrix. As a result, we can write \eqref{eqn:summuldef} as \begin{align}
    &I(\bm{Y};\bm{A}_1,\bm{A}_2) \notag \\
    &= \log_2\det\left(\bm{I}_{2MN}+\bm{\Sigma_\Omega}^{-1}  \bm{H}
    \left[\begin{matrix}\bm{\Sigma}_{\bm{A}1} &\bm{0} \\ \bm{0}  & \bm{\Sigma}_{\bm{A}2} \end{matrix}\right]
    \bm{H}^\dagger \right)    \\
    & = \log_2\det\left(\bm{I}_{2MN}+\bm{H}^\dagger\bm{\Sigma_\Omega}^{-1}  \bm{H}
    \left[\begin{matrix}\bm{\Sigma}_{\bm{A}1} &\bm{0} \\ \bm{0}  & \bm{\Sigma}_{\bm{A}2} \end{matrix}\right]
     \right) \label{eqn:mulinfodet}
\end{align}
where $\bm{\Sigma}_{\bm{A}k}$ is the covariance matrix for $\bm{A}_k$.

In order for \eqref{eqn:mulinfodet} to be valid, we need to make sure that the noise covariance matrix $\bm{\Sigma_\Omega}$ is invertible. We start by inspecting the positive definiteness of $\bm{\Sigma_\Omega}$. Notice that the quadratic form 
$\bm{A}^\dagger\bm{\Sigma_\Omega}\bm{A}$ 
is the sum of the energy of any $M$ signals, which are in the form $ \sum_{n=0}^{N-1}\left[a^m_1[n]p(t-n\delta T-\tau_1)+a^m_2[n]p(t-n\delta T-\tau_2)\right]$, $m=1,\dots,M$. We can see that as long as 
$\bm{A}$ is not a zero vector, the quadratic form will not be zero as the energy of the signal is not zero. Thus, we know that the noise covariance matrix is positive definitive and it is invertible.

\begin{remark}
For root raised cosine pulses assumed in this paper, we have to limit our discussion to $\delta\geq\frac{1}{1+\beta}$, in order to have numerical stability. Although $\bm{G}_{\delta}$ is always positive definite \cite{ourpaper}, some eigenvalues approach zero if $\delta < \frac{1}{1+\beta}$.
\end{remark} 

We can further simplify \eqref{eqn:mulinfodet} by manipulating the inverse of the noise covariance matrix in \eqref{eqn:noisecov}. Let us define
\begin{align}
\bm{Q}\triangleq \bm{G}_\delta-  \bm{G}_{\delta,12}\bm{G}_\delta^{-1}\bm{G}_{\delta,21}. \label{eqn:Q}
\end{align} To find the inverse of $\bm{\Sigma_\Omega}$ \eqref{eqn:noisecov}, we still need $\bm{Q}$ to be invertible. Therefore, we have the following lemma.
\begin{lemma}\label{lem:lem1}
The matrix $\bm{Q}$ in \eqref{eqn:Q} is invertible as long as $\delta\geq\frac{1}{1+\beta}$.
\end{lemma}
\begin{proof}
Notice that the matrix $\bm{Q}$ in \eqref{eqn:Q} is the Schur complement of $\left[\begin{matrix}
    \bm{G}_\delta & \bm{G}_{\delta,12} \\ \bm{G}_{\delta,21} & \bm{G}_\delta
\end{matrix}\right]$ \cite{matrices}. We know that the latter matrix is positive definite \cite{zhang2023capacity}, {if $\bm{G}_\delta$ is positive definite. Since the matrix $\bm{G}_\delta$ is always positive definite \cite{ourpaper}, by the Schur complement lemma \cite{matrices}, we conclude that $\bm{Q}$ is positive definite and thus invertible.} 
\end{proof} 

Then, using the definition of $\bm{Q}$ and Lemma~\ref{lem:lem1}, we can write the inverse of $\bm{\Sigma_\Omega}$ in \eqref{eqn:noisecov} as in \eqref{eqn:invnoisecov} {on the next page.} \begin{figure*}[t]
    \begin{align}
&\left[\begin{matrix}\bm{I}_M\otimes\bm{G}_\delta &\bm{I}_M\otimes\bm{G}_{\delta,12} \\ \bm{I}_M\otimes\bm{G}_{\delta,21}  & \bm{I}_M\otimes\bm{G}_\delta \end{matrix}\right]^{-1}\notag \\ &=
    \left[\begin{matrix}\makecell{\bm{I}_M\otimes\bm{G}_\delta^{-1} +(\bm{I}_M\otimes\bm{G}_\delta^{-1})(\bm{I}_M\otimes\bm{G}_{\delta,12})  \\ \times\big((\bm{I}_M\otimes\bm{G}_{\delta})-  (\bm{I}_M\otimes\bm{G}_{\delta,12})(\bm{I}_M\otimes\bm{G}_{\delta}^{-1}) \\ \times(\bm{I}_M\otimes\bm{G}_{\delta,21})\big)^{-1}(\bm{I}_M\otimes\bm{G}_{\delta,21})(\bm{I}_M\otimes\bm{G}_\delta^{-1}) }  &\makecell{-(\bm{I}_M\otimes\bm{G}_\delta^{-1})(\bm{I}_M\otimes\bm{G}_{\delta,12})\big((\bm{I}_M\otimes\bm{G}_{\delta})\\-   (\bm{I}_M\otimes\bm{G}_{\delta,12})(\bm{I}_M\otimes\bm{G}_{\delta}^{-1}) (\bm{I}_M\otimes\bm{G}_{\delta,21})\big)^{-1} }\\ &  \\ \makecell{-\big((\bm{I}_M\otimes\bm{G}_{\delta})-   (\bm{I}_M\otimes\bm{G}_{\delta,12})(\bm{I}_M\otimes\bm{G}_{\delta}^{-1}) \\ \times  (\bm{I}_M\otimes\bm{G}_{\delta,21})\big)^{-1}(\bm{I}_M\otimes\bm{G}_{\delta,21})(\bm{I}_M\otimes\bm{G}_\delta^{-1})  }  & \makecell{\big((\bm{I}_M\otimes\bm{G}_{\delta})-   (\bm{I}_M\otimes\bm{G}_{\delta,12})\\ \times (\bm{I}_M\otimes\bm{G}_{\delta}^{-1})  (\bm{I}_M\otimes\bm{G}_{\delta,21})\big)^{-1}} \end{matrix}\right] \notag\\
    & =\left[\begin{matrix}     {\bm{I}_M\otimes\big(\bm{G}_\delta^{-1}+\bm{G}_\delta^{-1}\bm{G}_{\delta,12}\bm{Q}^{-1}\bm{G}_{\delta,21}\bm{G}_\delta^{-1}\big)} & {-\bm{I}_M\otimes\big(\bm{G}_\delta^{-1}\bm{G}_{\delta,12}\bm{Q}^{-1}\big)} \\
        {-\bm{I}_M\otimes\big(\bm{Q}^{-1}\bm{G}_{\delta,21}\bm{G}_\delta^{-1}\big)}  &  {\bm{I}_M\otimes\bm{Q}^{-1}}
    \end{matrix}\right] \label{eqn:invnoisecov}
    \end{align}
\end{figure*}Moreover, we define 
    \begin{align}
   \left[\begin{matrix}\bm{M}_{11} & \bm{M}_{12} \\ \bm{M}_{21} & \bm{M}_{22}\end{matrix}\right] \triangleq \bm{H}^\dagger\bm{\Sigma_\Omega}^{-1}\bm{H} \label{eqn:M}
\end{align}
and by applying the mixed product property of Kronecker product, we obtain a simplification for each block of the matrix in \eqref{eqn:M} as
\begin{align}
\bm{M}_{11}&=\bm{H}_1^\dagger\bm{H}_1\otimes\bm{G}_\delta \label{eqn:simp1}\\ 
\bm{M}_{12}&=\bm{H}_1^\dagger\bm{H}_2\otimes\bm{G}_{\delta,12}  \\
\bm{M}_{21}&=\bm{H}_2^\dagger\bm{H}_1\otimes\bm{G}_{\delta,21} \\
\bm{M}_{22}&=\bm{H}_2^\dagger\bm{H}_2\otimes\bm{G}_\delta.  \label{eqn:simp2}
    \end{align}
Therefore, the simplified mutual information expression for the sum rate in \eqref{eqn:mulinfodet} becomes 
\begin{align}
    &I(\bm{Y};\bm{A}_1,\bm{A}_2)  
    =\log_2\det\bigg(\bm{I}_{2MN}  \notag  \\ & +   \left[\begin{matrix}\bm{H}_1^\dagger\bm{H}_1\otimes\bm{G}_\delta & \bm{H}_1^\dagger\bm{H}_2\otimes\bm{G}_{\delta,12} \\ \bm{H}_2^\dagger\bm{H}_1\otimes\bm{G}_{\delta,21} & \bm{H}_2^\dagger\bm{H}_2\otimes\bm{G}_\delta\end{matrix}\right]
    \left[\begin{matrix}\bm{\Sigma}_{\bm{A}1} &\bm{0} \\ \bm{0}  & \bm{\Sigma}_{\bm{A}2} \end{matrix}\right]
     \bigg). \label{eqn:simpsummul}
\end{align}
We can also obtain the single user mutual information expression from \eqref{eqn:simpsummul}, since 
\begin{align}
   \lefteqn{ I(\bm{Y};\bm{A}_1|\bm{A}_2)=} \\ &=h(\bm{Y}|\bm{A}_2) - h(\bm{Y}|\bm{A}_1, \bm{A}_2) \notag\\
    &=\log\det\left(\mathbb{E}\left[\bm{Y}\bm{Y}^\dagger\right]\right)|_{\bm{A}_2=0}-\log\det\left(\mathbb{E}\left[\bm{\Omega}\bm{\Omega}^\dagger\right]\right) \\
    &=\log_2\det\left(\bm{I}_{2MN} +
    \bm{\Sigma_\Omega}^{-1}\bm{H}\left[\begin{matrix}\bm{\Sigma}_{\bm{A}1} &\bm{0} \\ \bm{0}  & \bm{0} \end{matrix}\right]\bm{H}^\dagger
     \right)   \\
     &=\log_2\det\bigg(\bm{I}_{2MN} \notag\\ 
 &\quad\quad+\left[\begin{matrix}\bm{H}_1^\dagger\bm{H}_1\otimes\bm{G}_\delta & \bm{H}_1^\dagger\bm{H}_2\otimes\bm{G}_{\delta,12} \\ \bm{H}_2^\dagger\bm{H}_1\otimes\bm{G}_{\delta,21} & \bm{H}_2^\dagger\bm{H}_2\otimes\bm{G}_\delta\end{matrix}\right]\left[\begin{matrix}\bm{\Sigma}_{\bm{A}1} &\bm{0} \\ \bm{0}  & \bm{0} \end{matrix}\right]\bigg)\\
 &=\log_2\det\left(\bm{I}_{2MN}+\left(\bm{H}_1^\dagger\bm{H}_1\otimes\bm{G}_\delta\right)\bm{\Sigma}_{\bm{A}1}\right).\label{eqn:singleuserrate}
\end{align}
As expected, \eqref{eqn:singleuserrate} is the same as the MIMO FTN mutual information expression in \cite[(29)]{ourpaper}. Similarly, the single user rate expression for the second user is 
\begin{align}
    I(\bm{Y};\bm{A}_2|\bm{A}_1)=\log_2\det\left(\bm{I}_{2MN}+\left(\bm{H}_2^\dagger\bm{H}_2\otimes\bm{G}_\delta\right)\bm{\Sigma}_{\bm{A}2}\right).
\end{align}

Each user in a multiple access channel is limited in transmission power. Since users perform MIMO FTN transmission, the power constraint for each user is equal to the single user MIMO FTN power constraint \cite{ourpaper}, and is written as
\begin{equation}
    \frac{1}{N\delta T} \text{tr}\left(\left(\bm{I}_M\otimes\bm{G}_\delta\right)\bm{\Sigma}_{\bm{A}k}\right)\leq P_k,  \quad k=1,2. \label{eqn:powerconstraint}
\end{equation}
Thus, we can now write the capacity region for this channel according to \cite{zhang2023capacity} as 
\begin{equation}
    C=\text{closure}\left(\underset{N\rightarrow\infty}{\liminf}~C_N\right), \label{eqn:detregdef}
\end{equation}
where $C_N$ is\footnote{Note that the symbol $\succeq$ means positive definite.}
\begin{equation}
	\begin{aligned}
		C_N= \underset{\substack{\frac{1}{N\delta T} \text{tr}\left(\left(\bm{I}_M\otimes\bm{G}_\delta\right)\bm{\Sigma}_{\bm{A}k}\right)\leq P_k \\ \bm{\Sigma}_{\bm{A}k}\succeq0, ~k=1,2 }}{\bigcup}\bigg\{(R_1, R_2): \label{eqn:CN}
	\end{aligned}
\end{equation}
\begin{align}  &R_1\leq\log_2\det\left(\bm{I}_{2MN}+\left(\bm{H}_1^\dagger\bm{H}_1\otimes\bm{G}_\delta\right)\bm{\Sigma}_{\bm{A}1}\right) \label{eqn:single1rate}\\
		  &R_2\leq\log_2\det\left(\bm{I}_{2MN}+\left(\bm{H}_2^\dagger\bm{H}_2\otimes\bm{G}_\delta\right)\bm{\Sigma}_{\bm{A}2}\right) \label{eqn:single2rate}\\
		  &R_1+R_2\leq \log_2\det\bigg(\bm{I}_{2MN}  \notag \\    &+\left[\begin{matrix}\bm{H}_1^\dagger\bm{H}_1\otimes\bm{G}_\delta & \bm{H}_1^\dagger\bm{H}_2\otimes\bm{G}_{\delta,12} \\ \bm{H}_2^\dagger\bm{H}_1\otimes\bm{G}_{\delta,21} & \bm{H}_2^\dagger\bm{H}_2\otimes\bm{G}_\delta\end{matrix}\right]
    \left[\begin{matrix}\bm{\Sigma}_{\bm{A}1} &\bm{0} \\ \bm{0}  & \bm{\Sigma}_{\bm{A}2} \end{matrix}\right]
     \bigg)\bigg\}. \label{eqn:sumrate}  
\end{align}

\begin{remark}
    From \eqref{eqn:CN}-\eqref{eqn:sumrate}   we can see that when $\delta=1$, the $\bm{G}_\delta$ matrix becomes the identity matrix $\bm{I}_{N}$ and the $(n,m)th$ entries of $\bm{G}_{\delta,12}$ and $\bm{G}_{\delta,21}$ become {$g\left(\left(n-m\right)T+\tau_1-\tau_2\right)$ and $g\left(\left(n-m\right)T+\tau_2-\tau_1\right)$} respectively. The capacity region of the MIMO asynchronous MAC with FTN reduces to the MIMO asynchronous MAC capacity region without FTN \cite{verdu,ganji}. 
\end{remark}

\section{Proposed Achievable Rate Region}\label{sec:powallo}
Although \eqref{eqn:detregdef} is the capacity region for MIMO asynchronous MAC with FTN, it is not easy to calculate it in closed form or to optimize it. This is due to the fact that there is no single covariance matrix that satisfies both the single user rate expressions \eqref{eqn:single1rate} and \eqref{eqn:single2rate}, and the sum rate expression \eqref{eqn:sumrate} simultaneously. Therefore, the capacity region is not a pentagon as in synchronous MAC \cite{cover}, but has smooth corners as in \cite{verdu}.

 As the system includes MIMO, asynchronous transmission, and FTN, we propose a power allocation scheme to optimize over the MIMO channel, asynchronous transmission and FTN to obtain an achievable rate region. Thus, we suggest a set of input covariance matrices 
 $\mathcal{R}$. The set $\mathcal{R}$ consists of input covariance matrix pairs parameterized by $\alpha\in\mathbb{R}, 0\leq\alpha\leq1$,  namely, 
 \begin{equation}
     \mathcal{R}=\left\{ \left(\bm{\Sigma}^\alpha_{\bm{A}1},\bm{\Sigma}^\alpha_{\bm{A}2}\right)  \right\},
 \end{equation}where the covariance matrices $\bm{\Sigma}^\alpha_{\bm{A}k}, k=1,2,$ have the structure 
\begin{equation}
\bm{\Sigma}^\alpha_{\bm{A}k}=\bm{Z}_k\otimes\bm{\Xi}^\alpha_k. \label{eqn:Zk}
\end{equation}
The aim behind this structure is to adapt to the MIMO channel via $\bm{Z}_k$ and to provide precoding, or time correlation, against inter-symbol interference due to asynchronous transmission and FTN via $\bm{Xi}_k^\alpha$. 

In \eqref{eqn:Zk}, $\alpha$ is introduced as an auxiliary variable to obtain all rate pairs on the achievable rate region boundary. The $L \times L$  matrix $\bm{Z}_k$ is obtained by waterfilling \cite{telatarmimo} according to the MIMO channel $\bm{H}_k$ with the power constraint 
\begin{equation}
    \text{tr}\left(\bm{Z}_k\right)\leq P_k.
\end{equation}
The matrices $\bm{Z}_1, \bm{Z}_2$ remain the same for all $\alpha, 0\leq\alpha\leq1$. Thus they are the same for all $\bm{\Sigma}^\alpha_{\bm{A}1}$ and $\bm{\Sigma}^\alpha_{\bm{A}2}$, respectively. 


To obtain the $N \times N$ matrices $\bm{\Xi}^\alpha_k$, we first introduce the matrices $\bm{\Psi}_k^\alpha$ as
\begin{align} \bm{\Psi}^\alpha_1&=\bm{U}\tilde{\bm{\Psi}}^\alpha_1\bm{U}^\dagger=\bm{G}_\delta^{\frac{1}{2}}\bm{\Xi}^\alpha_1\bm{G}_\delta^{\frac{1}{2}} \label{eqn:Psidef}\\
\bm{\Psi}^\alpha_2&=\bm{V}\tilde{\bm{\Psi}}^\alpha_2\bm{V}^\dagger=\bm{G}_\delta^{\frac{1}{2}}\bm{\Xi}^\alpha_2\bm{G}_\delta^{\frac{1}{2}}, \label{eqn:Psidef2}
\end{align} where $\bm{U}$, and $\bm{V}$ are the unitary matrices in the singular value decomposition of 
\begin{equation}
    \bm{G}_\delta^{-\frac{1}{2}}\bm{G}_{\delta,12}\bm{G}_\delta^{-\frac{1}{2}}=\bm{U}\bm{\Lambda}\bm{V}^\dagger. \label{eqn:getlambda}
\end{equation} Here $\bm{\Lambda}$ is a diagonal matrix with the singular values of the matrix $\bm{G}_\delta^{-\frac{1}{2}}\bm{G}_{\delta,12}\bm{G}_\delta^{-\frac{1}{2}}$ on the diagonal. We denote those singular values by $\lambda_1, \dots, \lambda_N$. The matrices $\tilde{\bm{\Psi}}^\alpha_1$ and $\tilde{\bm{\Psi}}^\alpha_2$ are diagonal matrices, written as $\tilde{\bm{\Psi}}^\alpha_1=\text{diag}[\psi_1^1(\alpha),\dots,\psi_1^N(\alpha)]$ and $\tilde{\bm{\Psi}}^\alpha_2=\text{diag}[\psi_2^1(\alpha),\dots,\psi_2^N(\alpha)]$. 
The eigenvalues $\psi_k^i(\alpha), k=1,2, i=1,\dots,N$ are obtained from the mapping $f(\cdot)$, defined in \eqref{eqn:falpha} {on the next page.} $f(\cdot)$ takes in $\alpha$ as a parameter and returns a set of eigenvalues $\left[\psi_1^1(\alpha),\dots,\psi_1^N(\alpha), \psi_2^1(\alpha),\dots,\psi_2^N(\alpha)\right]$, or
\begin{equation} f(\alpha)=\left[\psi_1^1(\alpha),\dots,\psi_1^N(\alpha), \psi_2^1(\alpha),\dots,\psi_2^N(\alpha)\right].
\end{equation}
The mapping $f(\alpha)$ can be interpreted as  a different optimization problem for each $\alpha$ value and the solution to that specific problem is the desired set of eigenvalues $\left[\psi_1^1(\alpha),\dots,\psi_1^N(\alpha), \psi_2^1(\alpha),\dots,\psi_2^N(\alpha)\right]$\footnote{From this point, we will drop the argument of $\psi_k^i(\alpha)$ and write $\psi_k^i$ only for a simpler notation.}. For each $\alpha$, we plug the obtained eigenvalues back to  \eqref{eqn:Psidef},  \eqref{eqn:Psidef2} and then \eqref{eqn:Zk} to obtain a pair of input covariance matrices $\left(\bm{\Sigma}^\alpha_{\bm{A}1},\bm{\Sigma}^\alpha_{\bm{A}2}\right)$. We can then compute a rate pair $(R_1, R_2)$ on the achievable rate region boundary by plugging $\left(\bm{\Sigma}^\alpha_{\bm{A}1},\bm{\Sigma}^\alpha_{\bm{A}2}\right)$ into \eqref{eqn:single1rate}-\eqref{eqn:sumrate}. 

The power constraint for each optimization problem is equal to $\sum_{i=1}^N\psi_k^i\leq N\delta T, k=1,2$. This way, it is guaranteed that the individual power constraints of \eqref{eqn:powerconstraint} are satisfied for each user, since
\begin{align}
\text{tr}\left(\left(\bm{I}_M\otimes\bm{G}_\delta\right)\bm{\Sigma}^\alpha_{\bm{A}k}\right)&=\text{tr}\left(\left(\bm{I}_M\otimes\bm{G}_\delta\right)\left(\bm{Z}_k\otimes\bm{\Xi}^\alpha_k\right)\right) \notag\\  &=\text{tr}\left(\bm{Z}_k\otimes\bm{G}_\delta\bm{\Xi}^\alpha_k\right) \notag\\
     &=\text{tr}\left(\bm{Z}_k\right)\text{tr}\left(\bm{G}_\delta\bm{\Xi}^\alpha_k \right) \notag\\
     &=\text{tr}\left(\bm{Z}_k\right)\text{tr}\left(\bm{G}_\delta^{\frac{1}{2}}\bm{\Xi}^\alpha_k\bm{G}_\delta^{\frac{1}{2}} \right) \notag\\
     &\overset{(a)}{=}\text{tr}\left(\bm{Z}_k\right)\left(\sum_{i=1}^N\psi_k^i \right) \notag\\
     &\leq P_k N\delta T,\quad\quad\quad\quad k=1,2,
\end{align}
where (a) is because of \eqref{eqn:Psidef} and \eqref{eqn:Psidef2}.

\begin{figure*}
    \begin{equation}
        f(\alpha)=\left\{ \begin{aligned}          &\underset{\substack{\sum_{i=1}^N\psi_k^i\leq N\delta T,\\ \psi_i\geq0,i=1,\dots,N,k=1,2}}{\argmax}  \frac{1-2\alpha}{N}\sum_{i=0}^{N-1}\log_2\left(1+\frac{\psi_{2}^i}{\sigma_0^2}\right)+ \frac{\alpha}{N}\sum_{i=0}^{N-1}\log_2\left(1+\frac{\psi^i_{1}}{\sigma_0^2}+\frac{\psi_{2}^i}{\sigma_0^2}+\frac{\psi_{1}^i\psi_{2}^i}{\sigma_0^4}(1-|\lambda_i|^2)\right),  \notag \\ &\quad\quad\quad\quad\quad\quad\quad\quad\quad\quad\quad\quad\quad\quad\quad\quad\quad\quad\quad\quad\quad\quad\quad\quad\quad\quad\quad\quad\quad\quad\quad\quad\quad\quad\quad\quad\quad\text{if $0\leq \alpha\leq\frac{1}{2}$}  \\
       &\underset{\substack{\sum_{i=1}^N\psi_k^i\leq N\delta T,\\ \psi_i\geq0,i=1,\dots,N,k=1,2}}{\argmax} \frac{2\alpha-1}{N}\sum_{i=0}^{N-1}\log_2\left(1+\frac{\psi_{1}^i}{\sigma_0^2}\right)+ \frac{1-\alpha}{N}\sum_{i=0}^{N-1}\log_2\left(1+\frac{\psi^i_{1}}{\sigma_0^2}+\frac{\psi_{2}^i}{\sigma_0^2}+\frac{\psi_{1}^i\psi_{2}^i}{\sigma_0^4}(1-|\lambda_i|^2)\right), \notag \\ &\quad\quad\quad\quad\quad\quad\quad\quad\quad\quad\quad\quad\quad\quad\quad\quad\quad\quad\quad\quad\quad\quad\quad\quad\quad\quad\quad\quad\quad\quad\quad\quad\quad\quad\quad\quad\quad\text{if $\frac{1}{2}< \alpha\leq1$}
       \end{aligned} \right. \label{eqn:falpha}
    \end{equation}

\end{figure*}

\begin{remark}
    The intuition for the  optimization problem $f(\alpha)$ can be seen in the frequency domain. By applying Szeg\"o's theorem \cite{gray} to \eqref{eqn:Psidef}-\eqref{eqn:getlambda}, we obtain the spectrum domain representations of the matrices $\bm{\Psi}_1^\alpha$ and $\bm{\Psi}_2^\alpha$ respectively as $G_\delta(\lambda)S^\alpha_1(\lambda)$ and $G_\delta(\lambda)S^\alpha_2(\lambda)$. The function $G_\delta(\lambda)$ is the folded spectrum \cite{rusek}, which is defined as a function of the spectrum of the pulse shaping filter $P(\cdot)$; i.e.the continuous time Fourier transform of $p(t)$, as 
     \begin{equation}
   G_\delta(\lambda) = \frac{1}{\delta T}\sum_{n=-\infty}^{\infty}\left|P\left(\frac{\lambda-n}{\delta T}\right)\right|^2=\frac{1}{\delta T}\sum_{n=-\infty}^{\infty}G\left(\frac{\lambda-n}{\delta T}\right).\label{eqn:folded}
 \end{equation}
 The function $S^\alpha_k(\lambda)$ is the spectrum domain representation of the matrix $\bm{\Xi}_k^\alpha$. We also convert $\bm{G}_\delta^{-\frac{1}{2}}\bm{G}_{\delta,12}\bm{G}_\delta^{-\frac{1}{2}}$ into the spectrum domain to obtain $G_{12,\delta}(\lambda)/G_{\delta}(\lambda)$, where $G_{12,\delta}(\lambda)$ is defined as
 \begin{align}
  G_{12,\delta}(\lambda)&=\sum_{n=\infty}^{\infty}g(n\delta T + (\tau_1-\tau_2))e^{j2\pi\lambda n} \\
  &=\frac{1}{\delta T}\sum_{n=-\infty}^{\infty}G\left(\frac{\lambda-n}{\delta T}\right)e^{j2\pi(\tau_1-\tau_2)\frac{\lambda-n}{\delta T}}.
\end{align}
Given these definitions, we can say that solving $f(\alpha)$ is equivalent to finding the optimal input spectrum $G_\delta(\lambda)S^\alpha_k(\lambda), k=1,2$, at each one of the users according to the equivalent channel $G_{12,\delta}(\lambda)/G_{\delta}(\lambda)$. 
\end{remark} 
\begin{remark} We can show that the suggested power allocation scheme achieves the optimal single-user rate; i,e. the MIMO FTN capacity in \cite{ourpaper}. \end{remark}

To see this, we can look into $\alpha=0$ and $\alpha=1$. When $\alpha=0$, the optimization problem in \eqref{eqn:falpha} reduces to
\begin{equation}
    \underset{\substack{\sum_{i=1}^N\psi_k^i\leq N\delta T,\\ \psi_i\geq0,i=1,\dots,N,k=1,2}}{\argmax}\frac{1}{N}\sum_{i=0}^{N-1}\log_2\left(1+\frac{\psi_{2}^i}{\sigma_0^2}\right).
\end{equation}
The solution to this problem is $\psi_2^i=\delta T, i=1,\dots,N$. Then, $\bm{\Xi}_2^0 = \delta T\bm{G}^{-1}_\delta$, and the corresponding covariance matrix for the second user becomes $\bm{\Sigma}^0_{\bm{A}2}=\bm{Z}_2\otimes\left(\delta T\bm{G}_\delta^{-1}\right)$, which is the capacity-achieving covariance matrix for single-user MIMO FTN channel discussed in \cite{ourpaper}. When $\alpha=1$, the same discussion applies and we get the capacity-achieving covariance matrix for the first user.

\begin{figure}[t]
	\includegraphics[scale=0.6]{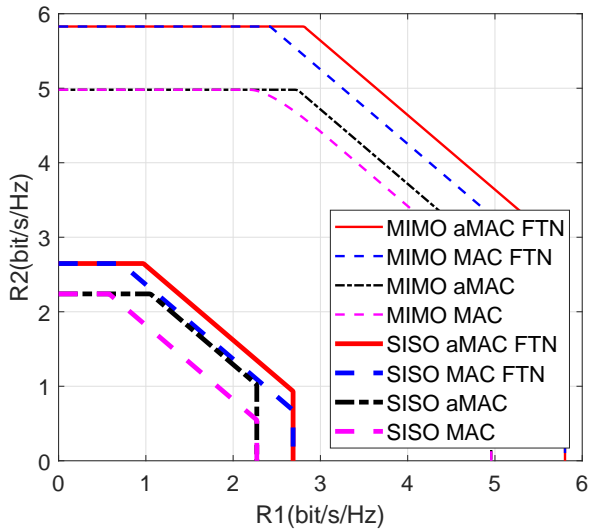}
	\caption{Achievable rate regions for MIMO asynchronous MAC with and without FTN, MIMO synchronous MAC with FTN, and the capacity region of MIMO MAC without FTN \cite{goldsmith}. Achievable rate regions for SISO transmission are also provided for comparison.}
	\label{fig:fig1}
\end{figure}
\begin{figure}[t]
	\includegraphics[scale=0.6]{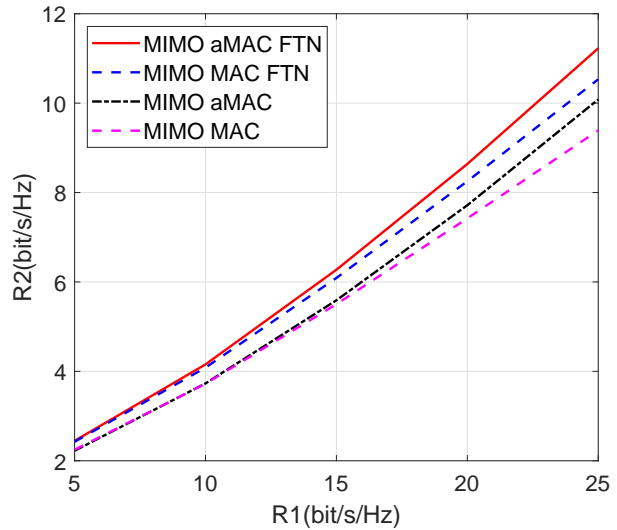}
	\caption{Maximum sum rate versus SNR for MIMO asynchronous MAC with and without FTN, MIMO synchronous MAC with FTN, and the sum capacity for MIMO MAC without FTN \cite{goldsmith}.}
	\label{fig:sumrate}
\end{figure}
\section{Numerical Results}\label{sec:num}
\begin{figure}[t]
	\includegraphics[scale=0.6]{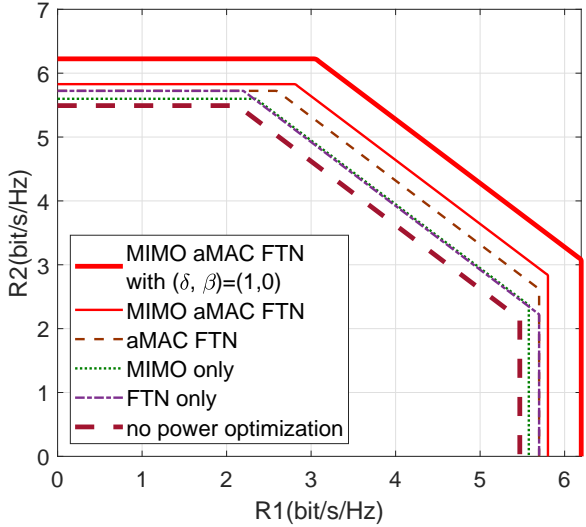}
	\caption{Achievable regions for MIMO aMAC FTN with different power allocation schemes.}
	\label{fig:fig2}
\end{figure}
In this section we plot the achievable rate region proposed in the previous section for different scenarios. We assume both the transmitters and the receiver have 3 antennas each. The transmitters employ root-raised cosine pulses for signaling and matched filtering. We set the symbol period $T=1$ and both users have signal to noise ratio (SNR)  $\frac{P_k}{\sigma^2_0}=20 \text{dB}, k=1,2$. All the regions are computed and averaged over the same 100 random channel realizations, where the power of the channel is normalized to 1. Furthermore, we set the time difference $\tau=\tau_2-\tau_1=0.5\delta T$ for the asynchronous transmission and we set $\delta=0.8$, $\beta=0.25$ for FTN. 

In Fig. \ref{fig:fig1}, we plot the achievable region of MIMO aMAC FTN and compare it with the achievable region of MIMO aMAC  to show the benefit of FTN. The MIMO aMAC region is obtained by performing waterfilling and asynchronous MAC power optimization.  Since there is no FTN for this curve, we obtain the $\bm{\Xi}_k^\alpha$ matrices in \eqref{eqn:Psidef}-\eqref{eqn:Psidef2} by replacing $\bm{G}_\delta$ with $\bm{I}_N$ and by evaluating $\bm{G}_{\delta,12}$ in \eqref{eqn:getlambda}  at $\delta=1$ respectively. We can see that FTN improves both the single-user rate as well as the sum-rate. In this figure, we also plot the MIMO MAC FTN achievable rate region by setting the time difference $\tau=0$, and the capacity region of MIMO MAC. We compare the MIMO aMAC region with MIMO MAC region, and also the MIMO aMAC FTN region with the MIMO MAC FTN region. Both of these comparisons show the gain from asynchronous transmission, either with or without FTN. We observe that asynchronous transmission improves sum rate but does not improve the single-user rate. In Fig. \ref{fig:fig1} we also plot the SISO counterparts for the above four regions and observe that gains due to asynchronous transmission and FTN are more emphasized in MIMO. 

To understand how the achievable rate regions change with SNR, in Fig.~\ref{fig:sumrate}, we plot the maximum sum rate versus SNR for MIMO asynchronous MAC with and without FTN, MIMO synchronous MAC with FTN and the sum capacity for MIMO MAC without FTN. 
We observe that asynchronous transmission with FTN starts presenting significant gains at moderate SNR values and this gain increases with SNR.

In order to understand how each component of the proposed power allocation scheme contributes individually, in Fig. \ref{fig:fig2}, we compare the MIMO aMAC FTN achievable rate region obtained with the power allocation scheme proposed in Section \ref{sec:powallo} with regions obtained by not omitting one or more components. The aMAC FTN region does not incorporate spatial power optimization and there is no waterfilling for the MIMO channels $\bm{H}_k$, $k=1,2$. The matrices $\bm{Z}_k$ of \eqref{eqn:Zk} are set to {$\bm{Z}_k= \frac{P_k}{L}\bm{I}_L$}, and we apply equal power allocation for each eigen-channel. On the other hand, the MIMO only curve means that we only perform waterfilling for the MIMO channels $\bm{H}_k$ and set the time difference $\tau=0$. We also do not perform power allocation for FTN, which means we replace   $\bm{G}_\delta$ and $\bm{G}_{\delta,12}$ in \eqref{eqn:Psidef}-\eqref{eqn:getlambda}  with $\bm{I}_N$ and $\bm{G}_{\delta,12}|_{\delta=1}$ respectively to obtain the $\bm{\Xi}_k^\alpha$ matrices in \eqref{eqn:Zk}. Similarly, the FTN only curve is obtained by performing equal power allocation for the MIMO channels and setting the time difference $\tau$ to be 0.  Finally, the no power optimization curve is obtained by performing equal power allocation for the MIMO channels $\bm{H}_k$, setting $\tau=0$  and   replacing  $\bm{G}_\delta$ and $\bm{G}_{\delta,12}$ in \eqref{eqn:Psidef}-\eqref{eqn:getlambda}  with $\bm{I}_N$ and $\bm{G}_{\delta,12}|_{\delta=1}$ respectively. We also plot the MIMO aMAC FTN with $(\delta,\beta)=(1,0)$ curve as an upper bound. As a rescdult, we conclude the following. If there are limited computational complexity resources to optimize power, we find that waterfilling for MIMO is a better approach to have reasonable performance. The MIMO only and FTN only curves have similar sum-rates, and it is easier to perform MIMO power allocation only. This is because larger matrices ($N \times N$) matrices need to be inverted for precoding against inter-symbol interference due to FTN, whereas $L \times L$ matrices are inverted to for waterfilling in space. However, if an increase in single-user rates is required, then FTN is helpful. Since the complexity for FTN only and aMAC FTN curves are similar, it is more meaningful to employ power optimization for asynchronous transmission with FTN together. This is because asynchronous transmission provides better inter-user interference mitigation and enhances the sum rate significantly. Obviously, it is the best if all power allocation mechanisms are utilized and the MIMO aMAC with FTN achieves the largest region among all.

\section{Conclusion}\label{sec:conc}

In this paper, we studied the performance of MIMO aMAC transmission with FTN. The system model is simplified by applying matrix manipulation techniques, then a novel power optimization scheme is proposed to obtain an achievable rate region. In the end, we show that our proposed power optimization scheme outperforms traditional MIMO MAC transmission and the gain brought by each part of the power optimization scheme.
Future work includes studying asynchronous transmission among antennas \cite{ovtdm,deliberate} and different acceleration factors among users. The extension of MIMO channel to massive MIMO will also be studied.

\bibliographystyle{IEEEtran}
\bibliography{main}

\end{document}